\documentclass[journal]{IEEEtran}
\usepackage{amsmath,amssymb,amsthm}
\usepackage{graphicx}
\usepackage{booktabs}
\usepackage{array}
\usepackage{algorithm}
\usepackage{algorithmic}
\usepackage[colorlinks=true,linkcolor=blue,citecolor=blue,urlcolor=blue]{hyperref}
\usepackage{tikz}
\usetikzlibrary{shapes.geometric, arrows.meta, positioning, calc}

\newtheorem{proposition}{Proposition}
\newtheorem{remark}{Remark}
\newcommand{\R}{\mathbb{R}}
\newcommand{\E}{\mathbb{E}}
\graphicspath{{../figs/}}

\begin{document}

\title{Certifying Collective Reasoning in Multi-Agent Systems via Koopman Spectral Analysis}

\author{Nuzhat~Khan and Indrakshi~Dey,~\IEEEmembership{Senior Member,~IEEE}%
\thanks{N.~Khan is with Universiti Teknologi Malaysia, Johor Bahru, Malaysia; I.~Dey is with Department of Computing and Mathematics, South East Technological University, Waterford, Ireland (e-mail: khan.nuzhat@utm.my; indrakshi.dey@setu.ie).}%

\thanks{Manuscript submitted to IEEE Transactions on Emerging Topics in Computational Intelligence and is under review.}}


\maketitle

\begin{abstract}
Orchestrated collectives of large language model (LLM) agents that debate and vote are an emerging form of computational intelligence: the intelligent behaviour resides in the \emph{interaction}, not in any single agent. They improve task accuracy, yet remain black boxes at the system level: there is no principled test of convergence, no bound on the rounds needed, and no faithful account of what drove a decision. This paper develops a novel framework based on Koopman operator theory and validates its theoretical guarantees on multi-agent consensus dynamics. Treating the collective as one nonlinear dynamical system on a communication graph, we read its essential behaviour off the spectrum of its Koopman transfer operator, an exact linear representation of the nonlinear dynamics estimated from interaction traces. The spectrum yields three machine-checkable certificates: the sub-dominant eigenvalue $\lambda_2$ fixes the intrinsic timescale of reasoning and yields a convergence deadline computable \emph{before} the debate runs; its eigenvector names the coherent factions the collective reasons in, and $|\lambda_2|$ certifies when that explanation is valid; and the leading spectral coordinates form a compressed, auditable message basis. On an attention-consensus model, the deadline tracks observed convergence with log--log correlation $0.93$ and bounds it in 96\% of 24 configurations; attribution is exact whenever the spectrum certifies metastability; eight of 32 coordinates preserve the decision at 99.7\% fidelity; and a certificate learned from 15 debates held on 60/60 held-out debates. The study runs in minutes on a CPU, making spectral certification a practical layer for trustworthy collective reasoning.
\end{abstract}

\begin{IEEEkeywords}
Collective intelligence, multi-agent systems, large language models, Koopman operator, dynamic mode decomposition, social reasoning, explainable artificial intelligence, consensus dynamics, trustworthy AI, semantic communications.
\end{IEEEkeywords}

\IEEEpeerreviewmaketitle

\section{Introduction}

\IEEEPARstart{A}{n} emerging class of intelligent systems locates its intelligence not in a single model but in a \emph{society} of models. Orchestrated large language model (LLM) agents that debate one another \cite{du2023}, criticise and defend candidate solutions \cite{liang2023}, or coordinate through structured conversation \cite{wu2023autogen} now routinely outperform their individual members on mathematical, coding, and question-answering benchmarks. In contrast to single-agent chain-of-thought prompting, where reasoning unfolds along an uncorrected autoregressive trajectory, multi-agent debate introduces iterative interaction and mutual feedback \cite{wei2026agentic}. This phenomenon, accuracy emerging from interaction, places LLM collectives squarely within the long-standing computational intelligence (CI) programme of studying how coordinated populations of simple or complex units produce capabilities that no unit possesses alone, a programme this journal has pursued across learning-based multi-agent consensus \cite{tetci_guang2025,tetci_zhang2025}, opponent and behaviour reasoning in agent teams \cite{tetci_hou2022}, trust formation in mixed agent teams \cite{tetci_lin2024}, and the dynamics of human--agent teaming \cite{tetci_demir2018}. What is new is the substrate: the ``opinions'' being exchanged are high-dimensional semantic embeddings produced by language models, the coupling between agents is an attention mechanism rather than a fixed gain, and the collective is deployed as a reasoning engine whose answers carry real-world consequences. Real-world deployment, moreover, requires such systems to operate in environments that are frequently novel, exhibiting phenomena no \emph{a priori} dataset anticipates and interactions that cannot be simulated in advance \cite{zhang2025data}. The consequence is a gap in what can be analysed, illustrated in Fig.~\ref{fig:conceptual_overview}: classical multi-agent consensus rests on an explicit graph topology and an analytical Lyapunov function (Fig.~\ref{fig:conceptual_overview}A), whereas a language-model collective evolves along natural-language trajectories whose update law is never written down (Fig.~\ref{fig:conceptual_overview}B), so that certifying its stability requires lifting observed trajectories into a space where the dynamics act linearly and reading the certificate off the resulting spectrum (Fig.~\ref{fig:conceptual_overview}C).

With deployment comes a demand that empirical benchmark gains cannot satisfy. Consider a practitioner who assembles a debate of eight LLM agents to answer safety-relevant questions. Three questions arise immediately, and none of them is answerable with current tools. \emph{(Q1) Convergence:} will this particular collective, on this communication topology, with these agents, settle on an answer at all, or will it oscillate, fragment into camps, or drift indefinitely? \emph{(Q2) Deadline:} if it does converge, how many rounds will that take? Every round costs inference compute, latency, and energy; a deployment must budget rounds, and today that budget is a guess. \emph{(Q3) Explanation:} when the collective commits to an answer, what drove it? The individual agents will spontaneously produce chain-of-thought narratives, but a growing body of evidence shows that such narratives are frequently \emph{unfaithful}, in the sense that they do not describe the computation that actually produced the answer \cite{jacovi2020,turpin2023}, and in any case they are accounts of individual agents, not of the collective process. Interpretability at the level of the system, as surveyed for neural networks in this journal \cite{tetci_zhang2021}, is precisely what is missing for neural \emph{societies}.

The central thesis of this paper is that all three questions become tractable, indeed become \emph{spectral} questions with computable answers, once the collective is viewed through the lens of Koopman operator theory. The idea, developed in detail in Section~\ref{sec:prelim}, can be stated in two sentences. Although the round-to-round evolution of the agents' belief embeddings is nonlinear (attention makes the coupling state-dependent), there exists a linear operator, the Koopman transfer operator \cite{koopman1931,mezic2005}, that represents this evolution \emph{exactly} on a suitable space of observable functions, and this operator can be approximated from recorded interaction traces using extended dynamic mode decomposition (EDMD) \cite{schmid2010,rowley2009,williams2015,schmid2022}. Once the operator is in hand, its spectrum does the rest: the gap between the trivial unit eigenvalue and the next eigenvalue $\lambda_2$ fixes the slowest decaying mode of disagreement and hence the number of rounds to consensus (Q1, Q2), while the eigenvector attached to $\lambda_2$ \emph{is} the pattern of disagreement the collective takes longest to resolve: the structural, system-level explanation of what the debate was actually about (Q3).

The purpose of this paper is to convert that thesis from an attractive analogy into a validated computational framework, and to be precise about what is demonstrated and what remains open. Our contributions are:

\begin{enumerate}
  \item \emph{An explanatory operator-theoretic framework for collective reasoning} (Sections~\ref{sec:prelim}--\ref{sec:method}). We develop, in tutorial style, the chain of reasoning that leads from a nonlinear attention-driven debate to a linear transfer operator, and we package the framework as a concrete certification pipeline (Algorithm~\ref{alg:pipeline}) whose inputs are nothing more than recorded interaction traces.
  \item \emph{A reference model of semantic debate} (Section~\ref{sec:model}): a minimal nonlinear attention-consensus system that isolates the two mechanisms that make LLM collectives dynamically nontrivial, namely similarity-gated (softmax) influence and graph-constrained communication, while remaining fully controllable and reproducible.
  \item \emph{Empirical validation of three certificates} (Section~\ref{sec:exp}). Across 24 collective configurations the spectral deadline tracks observed convergence with log--log correlation $r=0.93$ and is a sound upper bound in 96\% of cases (median conservatism $2.0\times$); spectral estimation error decays as $\approx M^{-0.36}$ over two decades of sample size; mode attribution recovers latent factions with 100\% accuracy in \emph{every} run in which the spectrum certifies that factions exist ($|\lambda_2|>0.9$; 42/60 runs), with $|\lambda_2|$ predicting attribution validity at correlation 0.92, so that the explanation is \emph{self-certifying}; spectral message compression retains 99.7\% decision fidelity at $4\times$ bandwidth reduction; and an end-to-end certificate learned from 15 training debates held on 60/60 held-out question-answering debates.
  \item \emph{A comparison with state-of-the-art baselines and a forward research programme} (Sections~\ref{sec:e9}, \ref{sec:disc}). Against classical graph-spectral consensus theory, empirical decay-curve fitting, and the fixed-round budgets used in current LLM-debate practice, the Koopman certificate is the only method that is simultaneously the most predictive (log--log $r=0.94$) and sound as a bound (100\% coverage); we then chart the research programme the framework opens, from finite-sample spectral theory to live LLM collectives.
\end{enumerate}

Everything reported here runs in under twenty CPU-minutes; code, traces, and results are released in full.
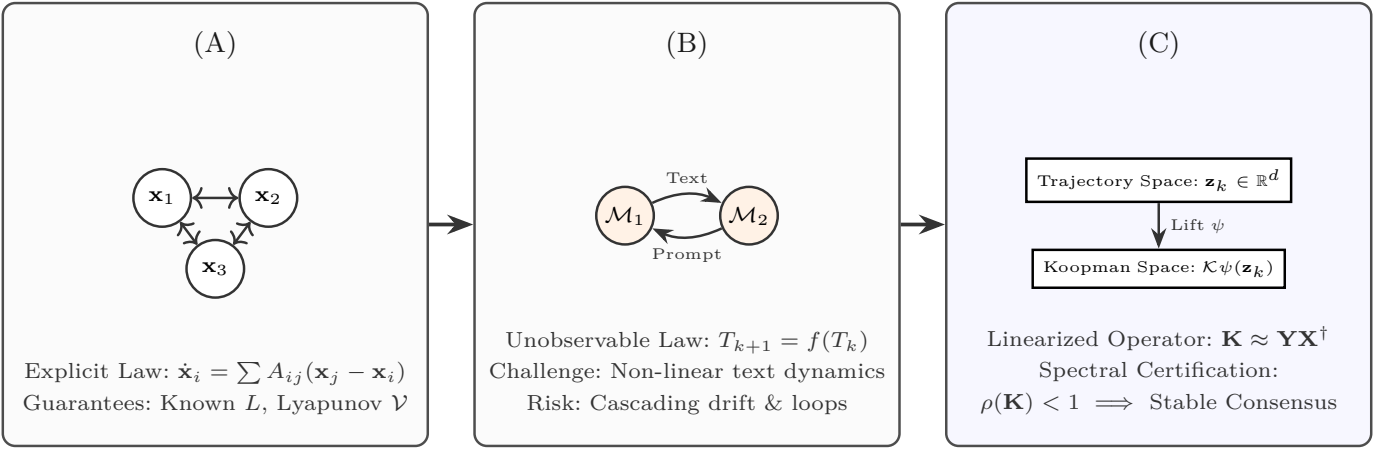
\begin{figure*}[t]
\centering
\resizebox{\textwidth}{!}{%
\begin{tikzpicture}[
    panel/.style={
        rectangle, 
        draw=black!70, 
        thick, 
        fill=gray!3, 
        rounded corners=4pt, 
        minimum width=4.8cm, 
        minimum height=5.0cm, 
        align=center
    },
    node_box/.style={
        circle, 
        draw=black!80, 
        thick, 
        fill=white, 
        minimum size=0.65cm, 
        font=\sffamily\bfseries\scriptsize,
        inner sep=0pt
    },
    arrow/.style={
        -Stealth, 
        thick, 
        color=black!80
    },
    title_text/.style={
        font=\sffamily\bfseries\small, 
        color=black!90
    },
    body_text/.style={
        font=\sffamily\scriptsize, 
        color=black!80, 
        align=center
    }
]

\node (panel1) [panel] {};
\node (p1_title) [title_text, anchor=north] at ($(panel1.north)+(0,-0.2)$) {(A)};
\node (a1) [node_box] at ($(panel1.center)+(-0.6,0.3)$) {$\mathbf{x}_1$};
\node (a2) [node_box] at ($(panel1.center)+(0.6,0.3)$) {$\mathbf{x}_2$};
\node (a3) [node_box] at ($(panel1.center)+(0,-0.5)$) {$\mathbf{x}_3$};
\draw [arrow, <->] (a1) -- (a2);
\draw [arrow, <->] (a1) -- (a3);
\draw [arrow, <->] (a2) -- (a3);
\node (p1_desc) [body_text, anchor=south] at ($(panel1.south)+(0,0.25)$) {
    Explicit Law: $\dot{\mathbf{x}}_i = \sum A_{ij}(\mathbf{x}_j - \mathbf{x}_i)$\\[2pt]
    \textbf{Guarantees:} Known $L$, Lyapunov $\mathcal{V}$
};

\node (panel2) [panel, right=0.5cm of panel1] {};
\node (p2_title) [title_text, anchor=north] at ($(panel2.north)+(0,-0.2)$) {(B)};
\node (l1) [node_box, fill=orange!10] at ($(panel2.center)+(-0.7,0.1)$) {$\mathcal{M}_1$};
\node (l2) [node_box, fill=orange!10] at ($(panel2.center)+(0.7,0.1)$) {$\mathcal{M}_2$};
\draw [arrow, bend left=25] (l1) to node[above, font=\tiny] {Text} (l2);
\draw [arrow, bend left=25] (l2) to node[below, font=\tiny] {Prompt} (l1);
\node (p2_desc) [body_text, anchor=south] at ($(panel2.south)+(0,0.25)$) {
    Unobservable Law: $T_{k+1} = f(T_k)$\\[2pt]
    \textbf{Challenge:} Non-linear text dynamics\\[2pt]
    \textbf{Risk:} Cascading drift \& loops
};

\node (panel3) [panel, right=0.5cm of panel2, fill=blue!3] {};
\node (p3_title) [title_text, anchor=north] at ($(panel3.north)+(0,-0.2)$) {(C)};
\node (state_space) [draw, rectangle, fill=white, thick, font=\tiny, align=center] at ($(panel3.center)+(0,0.5)$) {Trajectory Space: $\mathbf{z}_k \in \mathbb{R}^d$};
\node (koopman_space) [draw, rectangle, fill=white, thick, font=\tiny, align=center] at ($(panel3.center)+(0,-0.5)$) {Koopman Space: $\mathcal{K}\psi(\mathbf{z}_k)$};
\draw [arrow] (state_space) -- node[right, font=\tiny] {Lift $\psi$} (koopman_space);
\node (p3_desc) [body_text, anchor=south] at ($(panel3.south)+(0,0.25)$) {
    Linearized Operator: $\mathbf{K} \approx \mathbf{Y}\mathbf{X}^\dagger$\\[2pt]
    \textbf{Spectral Certification:}\\[2pt]
    $\rho(\mathbf{K}) < 1 \implies$ \textbf{Stable Consensus}
};

\draw [arrow, line width=1.2pt] (panel1.east) -- (panel2.west);
\draw [arrow, line width=1.2pt] (panel2.east) -- (panel3.west);

\end{tikzpicture}%
}
\caption{What is observable differs sharply between the two settings. In (A), agents exchange numerical states over a known graph, so the update law is explicit and convergence follows from a Lyapunov argument. In (B), the agents are language models and the round-to-round map $T_{k+1}=f(T_k)$ is never available in closed form. In (C), we recover it from data: trajectories are lifted through observables $\psi$, the operator $\mathbf{K}$ is estimated by least squares, and, once centering removes the trivial consensus direction, $\rho(\mathbf{K})<1$ certifies that every disagreement mode decays.}
\label{fig:conceptual_overview}
\end{figure*}

\section{Related Work}\label{sec:related}

\subsection{Multi-Agent Reasoning With Language Models}

Multi-agent LLM frameworks coordinate several model instances to solve complex tasks through structured debate, mutual critique, and consensus building \cite{tran2025multiagent}. Multi-agent debate \cite{du2023} lets several LLM instances answer independently and then revise in light of one another's answers; adversarial variants assign explicit critic roles \cite{liang2023}; and conversational orchestration frameworks such as AutoGen \cite{wu2023autogen} and LangGraph \cite{pelluru2025langchain} generalise these patterns to arbitrary role graphs, assign specialised agent roles, and standardise message exchange. Peer debate and evaluation improve reasoning accuracy by letting agents identify and correct errors early, and the consistent finding is that interaction helps; but the benefit depends sensitively on the number of rounds, the sparsity of the communication pattern, and the diversity of the participants, none of which current theory explains. What these frameworks offer in place of such a theory are rudimentary stopping rules, fixed turn limits, majority voting, or string-agreement checks, none of which models how the debate itself evolves \cite{motger2026multi}. They can therefore neither guarantee convergence, nor prevent repetition, nor validate the correctness of a consensus, which is why reliable assessment of a collective calls for methods that monitor and certify dialogue stability as it unfolds rather than relying on benchmark performance alone. Our work is complementary to this literature: we take the orchestrated collective as given and ask what can be \emph{guaranteed} about it from its own interaction traces.

\subsection{Learning-Based Consensus and Agent Teaming}
Learning-based consensus for multi-agent systems has been advanced through reinforcement-learning strategies for prescribed-time optimal consensus under switching stochastic dynamics \cite{tetci_guang2025} and through event-triggered impulsive consensus under time-varying delays \cite{tetci_zhang2025}; both lines certify convergence for \emph{designed} protocols with known structure, and in doing so depend on an explicit mathematical representation of the agents' update law. The complementary problem we address is certification for \emph{emergent} protocols: collectives whose update rule (an LLM conditioned on peers' messages) is unknown and can only be observed. The central challenge is therefore not to design a provably stable protocol but to certify the convergence of an unknown consensus process purely from observed interaction sequences, an emphasis that connects to recent progress in behavioural inference and dynamic trust modelling within agent teams. Related work on reasoning about other agents' behaviour \cite{tetci_hou2022} and on trust formation in agent teams \cite{tetci_lin2024} likewise extracts structure from interaction data rather than from models, using offline system identification to assess how agents align over time; language-model orchestration is now beginning to apply data-driven operator methods in the same spirit, recovering latent state transitions from text logs. Our contribution can be read as extending that data-driven philosophy from the level of individual agents to the level of the collective's dynamical law. Finally, studies of human--autonomous-agent teaming in \cite{tetci_demir2018} document how team-level dynamics, not individual competence, determine outcomes, a principle that modern teaming paradigms reinforce: static role assignments often fail once agent interactions produce deadlocks or repetitive loops under novel task conditions. The same premise is established here for machine-only teams, where tracing state transitions across communication rounds with operator-theoretic tools yields real-time stability guarantees that static benchmarks cannot supply.

\subsection{Koopman Operator Theory and Data-Driven Spectral Analysis}

The Koopman operator \cite{koopman1931,mezic2005} recasts nonlinear dynamics through an infinite-dimensional linear operator acting on the space of all measurement functions of the system. Framing nonlinear behaviour this way opens the door to using well-established linear-systems tools, developed over decades for prediction, estimation, and control, on problems that are fundamentally nonlinear; the price is infinite dimensionality, and the difficulty that remains is finding finite-dimensional coordinates or embeddings in which this linearity actually holds in practice. Dynamic mode decomposition (DMD) \cite{schmid2010,rowley2009,schmid2022} and its extended form EDMD \cite{williams2015} address precisely that difficulty, approximating the spectrum from snapshot data, with convergence guarantees in the infinite-data, rich-dictionary limit \cite{korda2018} and learned-dictionary variants that remove the burden of choosing observables by hand \cite{li2017edmddl,klus2020}.

Three factors explain why Koopman-based analysis has gained traction. First, it rests on solid theoretical foundations linking it back to established geometric treatments of dynamical systems. Second, because the framework is built around measurements rather than governing equations, it fits naturally with data-driven and machine-learning approaches. Third, practical numerical tools, most notably dynamic mode decomposition and its variants, have turned the underlying theory into something usable for real applications \cite{brunton2021}. These tools have certified and controlled fluid, robotic, and power systems. To our knowledge, the present paper is the first to apply the transfer-operator viewpoint to \emph{semantic} collectives, that is, systems whose state is a belief embedding and whose coupling is attention, and the first to derive reasoning-time and explanation certificates from the estimated spectrum.

\subsection{Consensus Theory, Opinion Dynamics, and Faithful Explanation}
For linear consensus, convergence rate is governed by the algebraic connectivity of the communication graph \cite{olfatisaber2004,degroot1974}; nonlinear bounded-confidence models exhibit metastable opinion clusters \cite{hegselmann2002}. Attention consensus, defined below, interpolates between these regimes and inherits the analytical difficulty of both: the effective coupling matrix depends on the state, so no fixed graph spectrum applies, which is exactly why a \emph{data-driven} operator spectrum is required. On the explanation side, the unfaithfulness of post-hoc rationales \cite{jacovi2020,turpin2023} motivates explanations that are properties of the dynamics rather than narratives about them; the neural-network interpretability taxonomy of \cite{tetci_zhang2021} distinguishes passive post-hoc analysis from explanations with verifiable semantics, and the mode attribution developed here belongs firmly to the latter class. The compression certificate connects, in turn, to diffusion-map coordinates \cite{coifman2006} and to goal-oriented semantic communication, where only decision-relevant message content is transmitted.

\section{Preliminaries: Why a Linear Operator Can Certify a Nonlinear Debate}\label{sec:prelim}

Because the Koopman viewpoint may be unfamiliar to parts of the CI community, we develop it here in explanatory terms; the expert reader may skim to Section~\ref{sec:model}.

\subsection{From States to Observables}
Consider any discrete-time dynamical system $z(t+1)=F(z(t))$ with state $z\in\R^{n}$ and $F$ nonlinear; for us, $z$ will collect all agents' belief embeddings and $F$ will encode one debate round. The conventional viewpoint tracks the \emph{state}. The Koopman viewpoint tracks \emph{measurements of the state}: any scalar function $\psi:\R^n\to\R$ (an ``observable'', for example the disagreement between two particular agents, or any nonlinear feature of the whole configuration). The dynamics act on observables by composition,
\begin{equation}
(\mathcal K\psi)(z) \;=\; \psi(F(z)),
\end{equation}
and for stochastic dynamics by $(\mathcal K\psi)(z)=\E[\psi(F(z))]$. The key, and at first sight surprising, fact is that $\mathcal K$ is a \emph{linear} operator regardless of how nonlinear $F$ is: for any observables $\psi_1,\psi_2$ and scalars $a,b$, $\mathcal K(a\psi_1+b\psi_2) = a\,\mathcal K\psi_1 + b\,\mathcal K\psi_2$, simply because composition distributes over addition of functions. Nothing is approximated; the price paid is that $\mathcal K$ acts on an infinite-dimensional space of functions rather than on $\R^n$.

\subsection{Spectrum as Timescales, Eigenfunctions as Patterns}
Why is linearity worth that price? Because linear operators have spectra, and spectra encode timescales \cite{liu2026spectral}. Suppose $\varphi$ is an eigenfunction, $\mathcal K\varphi = \lambda\varphi$. Then along any trajectory, $\varphi(z(t)) = \lambda^{t}\varphi(z(0))$: the scalar signal obtained by evaluating $\varphi$ evolves by pure geometric decay (or growth, or rotation, for complex $\lambda$), \emph{exactly}, forever, no matter how complicated $F$ is. An eigenvalue with $|\lambda|<1$ is therefore a decay channel with half-life $\ln 2/(-\ln|\lambda|)$ rounds, and its eigenfunction is the \emph{pattern in state space that decays at that rate}. For a debating collective, the interpretation is direct: the eigenvalue $\lambda=1$ corresponds to what the collective preserves (the consensus it is heading to); the sub-dominant eigenvalue $\lambda_2$ is the slowest-dying pattern of disagreement, and its modulus fixes how many rounds the collective needs before that disagreement falls below any tolerance. The quantity
\begin{equation}
\gamma \;=\; 1-|\lambda_2|
\end{equation}
is the \emph{spectral gap}; a large gap means fast collective convergence, a vanishing gap means metastability, that is, persistent factions \cite{shahbazi2025opinion}. Everything this paper certifies is a corollary of estimating $\lambda_2$ and its eigenfunction well.

\emph{A worked micro-example.} The smallest instructive case is two agents averaging scalar beliefs,
\begin{equation}\label{eq:micro}
  x_i(t{+}1)=(1-\alpha)\,x_i(t)+\alpha\,\bar x(t).
\end{equation}
The mean $\bar x$ is conserved (eigenvalue $1$), while the disagreement $\delta=x_1-x_2$ obeys $\delta(t{+}1)=(1-\alpha)\,\delta(t)$: the observable $\delta$ is a Koopman eigenfunction with eigenvalue $\lambda_2 = 1-\alpha$, so the collective halves its disagreement every $\ln 2/(-\ln(1-\alpha))$ rounds, a deadline readable from the spectrum. Attention ($\beta>0$) destroys the closed form by making the effective $\alpha$ state-dependent, but the \emph{structure} of the argument survives: a slowest disagreement observable still exists, and EDMD's task is to find it and its decay rate from data. Table~\ref{tab:notation} collects the notation used throughout.

\begin{table}[t]
\centering
\caption{Principal notation}
\label{tab:notation}
\begin{tabular}{ll}
\toprule
Symbol & Meaning\\
\midrule
$N,\;d$ & number of agents; belief-embedding dimension\\
$x_i(t)\in\R^d$ & belief embedding of agent $i$ at round $t$\\
$z(t)\in\R^{Nd}$ & collective state $\mathrm{vec}(x_1,\dots,x_N)$\\
$\delta(t)$ & centered deviation $z(t)-\bar z(t)$\\
$D(t)$ & normalised disagreement\\
$\beta,\;\alpha,\;\sigma$ & attention temperature; update rate; noise scale\\
$\mathcal K,\;\hat K$ & Koopman operator; its EDMD estimate\\
$\lambda_2,\;\gamma$ & sub-dominant eigenvalue; spectral gap $1-|\lambda_2|$\\
$\Psi,\;m,\;\rho$ & dictionary; \# random features; ridge parameter\\
$M,\;S$ & \# training trajectories; \# snapshot pairs\\
$\epsilon,\;T_{\mathrm{cert}}$ & disagreement tolerance; certified deadline\\
\bottomrule
\end{tabular}
\end{table}

\begin{figure*}[t]
\centering
\begin{tikzpicture}[
    box/.style={
        rectangle, 
        draw=black!80, 
        thick, 
        fill=blue!3, 
        rounded corners=3pt, 
        minimum width=3.4cm, 
        minimum height=2.2cm, 
        align=center, 
        font=\sffamily\small
    },
    arrow/.style={
        -Stealth, 
        thick, 
        color=black!80
    },
    sublabel/.style={
        font=\sffamily\scriptsize\color{black!70},
        below=2pt
    }
]

\node (b1) [box] {
    \textbf{1. Dialogue Traces}\\[6pt]
    \scriptsize Agent Execution Logs\\[4pt]
    $T_k \to T_{k+1}$
};

\node (b2) [box, right=0.8cm of b1] {
    \textbf{2. Semantic States}\\[6pt]
    \scriptsize Embedding Trajectory\\[4pt]
    $\mathbf{z}_k = E(T_k) \in \mathbb{R}^d$
};

\node (b3) [box, right=0.8cm of b2] {
    \textbf{3. EDMD Engine}\\[6pt]
    \scriptsize Operator Linearization\\[4pt]
    $\mathbf{K} \approx \mathbf{Y}\mathbf{X}^\dagger$
};

\node (b4) [box, right=0.8cm of b3] {
    \textbf{4. Stability Bounds}\\[6pt]
    \scriptsize Spectral Certification\\[4pt]
    $\rho(\mathbf{K}) < 1$
};

\draw [arrow] (b1) -- (b2);
\draw [arrow] (b2) -- (b3);
\draw [arrow] (b3) -- (b4);

\node [sublabel] at (b1.south) [anchor=north, yshift=-4pt] {Unstructured Text};
\node [sublabel] at (b2.south) [anchor=north, yshift=-4pt] {Continuous Vector Space};
\node [sublabel] at (b3.south) [anchor=north, yshift=-4pt] {Linear Matrix Estimation};
\node [sublabel] at (b4.south) [anchor=north, yshift=-4pt] {Stability Certified};

\end{tikzpicture}

\caption{The four stages of the certification pipeline. Agent logs $T_k$ are embedded into a semantic state $\mathbf{z}_k=E(T_k)$, turning a transcript into a trajectory; snapshot pairs from that trajectory are lifted and regressed to give the operator $\mathbf{K}\approx\mathbf{Y}\mathbf{X}^\dagger$; and its eigenvalues supply the certificates, with $\rho(\mathbf{K})<1$ establishing that disagreement decays and the sub-dominant eigenvalue fixing how fast. Everything after the embedding is linear algebra.}
\label{fig:pipeline_framework}
\end{figure*}
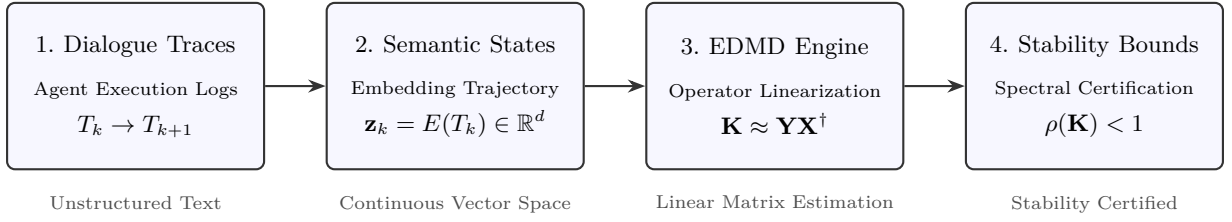

\subsection{Estimating the Operator From Traces: EDMD}
Figure~\ref{fig:pipeline_framework} summarises the complete computational pipeline that follows from this viewpoint; the present subsection develops its central estimation step. The operator is infinite-dimensional, but its action on a finite \emph{dictionary} of observables $\Psi(z)=[\psi_1(z),\dots,\psi_P(z)]$ can be estimated by regression. Given snapshot pairs $\{(z_s, z_s^{+})\}_{s=1}^{S}$ harvested from trajectories, EDMD \cite{williams2015} finds the matrix $\hat K\in\R^{P\times P}$ that best propagates dictionary values one step forward,
\begin{equation}\label{eq:edmd}
\hat K \;=\; \arg\min_{K}\; \sum_{s}\big\lVert \Psi(z_s^{+}) - \Psi(z_s)\,K \big\rVert^2 + \rho\lVert K\rVert_F^2 ,
\end{equation}
solved in closed form as $\hat K = (\Psi^\top\Psi+\rho I)^{-1}\Psi^\top\Psi^{+}$. The eigenvalues of $\hat K$ approximate Koopman eigenvalues; the quality of the approximation depends on how well the dictionary spans the relevant eigenfunctions and on how much data is available; both are probed empirically in Section~\ref{sec:e2}. Two practical points deserve emphasis for the collective-reasoning setting. First, the consensus manifold contributes trivial eigenvalues at $1$; \emph{centering} the state (subtracting the mean belief) before applying the dictionary removes this subspace and exposes the decay modes we care about. Second, the dictionary need not be hand-designed for each system: a generic dictionary of random Fourier features, consisting of linear coordinates plus $\cos(W_{\!f}z+b)$ with random $W_{\!f},b$, approximates a universal kernel space and suffices for the dynamics studied here, while learned dictionaries \cite{li2017edmddl,klus2020} are the natural upgrade path for real LLM embeddings.

\section{A Reference Model of Semantic Debate}\label{sec:model}

To validate certificates one needs a system whose ground truth is accessible: convergence times must be measurable, faction structure must be plantable, and thousands of runs must be affordable. We therefore construct the \emph{minimal} model that retains the two mechanisms responsible for the nontrivial dynamics of LLM collectives, and we are explicit that it is a model: the transfer to real LLM traces is discussed, with its risks, in Section~\ref{sec:disc}.

\subsection{Attention-Consensus Dynamics}
Let $N$ agents hold belief embeddings $x_i(t)\in\R^d$, the model analogue of the semantic embedding of an agent's current message, and let them communicate along a graph $G$ with closed neighbourhoods $\mathcal{N}(i)\cup\{i\}$. Each round has two phases, mirroring an LLM debate round. In the \emph{appraisal} phase, agent $i$ assigns each incoming message a weight that grows with its semantic agreement with the agent's own current belief:
\begin{equation}\label{eq:att}
  a_{ij}(t) = \frac{\exp\!\big(\beta\,\langle \hat x_i(t), \hat x_j(t)\rangle\big)}{\sum_{k\in\mathcal N(i)\cup\{i\}}\exp\!\big(\beta\,\langle \hat x_i(t), \hat x_k(t)\rangle\big)},
\end{equation}
where $\hat x = x/\lVert x\rVert$ and $\beta\ge0$ is an attention temperature. In the \emph{update} phase, the agent moves its belief a fraction $\alpha\in(0,1]$ of the way toward the attention-weighted mixture:
\begin{equation}\label{eq:dyn}
  x_i(t{+}1) = (1-\alpha)\,x_i(t) + \alpha\!\!\sum_{j\in\mathcal N(i)\cup\{i\}}\!\! a_{ij}(t)\,x_j(t) + \sigma\,\xi_i(t),
\end{equation}
with optional Gaussian process noise $\xi_i$ of scale $\sigma\ge0$. Three properties should be noted. (i) For $\beta=0$ the model reduces to classical linear consensus, whose rate is fixed by the graph Laplacian \cite{olfatisaber2004}; for $\beta>0$ the coupling matrix depends on the state, the map is genuinely nonlinear, and like-minded agents amplify one another, which is the mechanism by which LLM agents preferentially adopt arguments resembling their own. (ii) The collective state $z(t)=\mathrm{vec}(x_1,\dots,x_N)\in\R^{Nd}$ leaves the consensus manifold $\{x_1=\cdots=x_N\}$ invariant under the deterministic dynamics ($\sigma=0$), so ``an answer'' is an attractor family rather than a single point. (iii) Convergence is naturally monitored by the normalised disagreement
\begin{equation}
D(t)=\lVert z(t)-\bar z(t)\rVert \,/\, \lVert z(0)-\bar z(0)\rVert,
\end{equation}
where $\bar z$ stacks $N$ copies of the mean belief; $D$ is the model analogue of ``how far the debate still is from a common answer.''

\subsection{Question-Answering Variant}
To connect disagreement decay to task performance, we also study a decision-making variant. Agents hold logits $L_i(t)\in\R^{K}$ over $K$ answer options; attention weights \eqref{eq:att} are computed on the induced probability vectors; logits are mixed as in \eqref{eq:dyn}. A fraction $c$ of ``competent'' agents receives a weak evidence bias $\varepsilon_{\mathrm{ev}}$ toward the correct option at $t=0$, modelling the realistic situation where some, but not all, agents have weak private signal. The collective answer at round $t$ is the majority of $\arg\max_k L_i(t)$. This variant lets us ask the deployment-relevant question directly: \emph{by which round is the collective's answer final?}

\subsection{Faction Initialisation for Attribution Studies}
Explanations only have content when there is structure to explain. For the attribution experiments we therefore plant structure: agents are initialised near one of two well-separated belief centres, and the communication graph is a planted-partition (stochastic block model) graph with intra-/inter-faction edge probabilities $p_{\mathrm{in}}\gg p_{\mathrm{out}}$, reflecting the homophilous interaction patterns that attention itself induces. The ground-truth faction labels are used only for evaluation, never by the certification pipeline.

\section{The Certification Pipeline}\label{sec:method}

\subsection{Estimator}
We instantiate EDMD as follows. Each recorded trace is centered, $\delta(t)=z(t)-\bar z(t)$, removing the trivial consensus eigenspace. The dictionary is
\begin{equation}\label{eq:dict}
  \Psi(\delta) = \big[\,1,\;\; \delta,\;\; \cos(W_{\!f}\,\delta + b)\,\big] \in \R^{1+Nd+m},
\end{equation}
with $m$ random Fourier features ($W_{\!f}$ Gaussian with scale $0.7/\sqrt{Nd}$, $b\sim\mathcal U[0,2\pi]$); $\hat K$ is obtained from \eqref{eq:edmd} with ridge $\rho=10^{-6}$. Eigenvalues are sorted by modulus; residual near-unit eigenvalues ($|\lambda|\ge 0.985$), which correspond to conserved quantities rather than decay, are excluded, and $\lambda_2$ denotes the largest remaining modulus. The computational cost is that of one regularised least-squares problem, $O(P^2 S + P^3)$ for dictionary size $P$ and $S$ snapshot pairs, which is negligible next to a single LLM inference call, which is what makes certification a plausible always-on layer.

\subsection{Certificate 1: A Convergence Deadline}
To turn the spectral properties of the learned operator $\hat K$ into an operational metric, we define a predicted convergence deadline $T_{\mathrm{pred}}(\gamma)$ as an explicit function of the spectral gap, so that what would otherwise be heuristic multi-agent execution becomes a time-bounded process with a stated maximum number of interaction rounds. If the centered dynamics are dominated by the $\lambda_2$-mode, disagreement decays geometrically, $D(t)\le C\,|\lambda_2|^{t}$, and solving for the tolerance $\epsilon$ gives the deadline
\begin{equation}\label{eq:deadline}
  T_{\mathrm{pred}}(\gamma) \;=\; \frac{\ln(C/\epsilon)}{-\ln|\lambda_2|},\qquad T_{\mathrm{cert}}=\lceil T_{\mathrm{pred}}\rceil .
\end{equation}
Throughout we use the worst-case constant $C=1$ and $\epsilon=0.05$. The following proposition makes the underlying assumption explicit rather than hiding it.

\begin{proposition}[Deadline soundness under spectral dominance]\label{prop:deadline}
Suppose the centered deviation admits an expansion $\delta(t)=\sum_{j\ge 2} c_j \lambda_j^{\,t}\, v_j$ in Koopman modes with $|\lambda_2|\ge|\lambda_3|\ge\cdots$ and $\sum_{j\ge2}|c_j|\lVert v_j\rVert \le C\,\lVert\delta(0)\rVert$. Then $D(t)\le C\,|\lambda_2|^{t}$, and $D(t)\le\epsilon$ for every $t\ge T_{\mathrm{pred}}(\gamma)$.
\end{proposition}
\begin{proof}
$\lVert\delta(t)\rVert \le \sum_{j\ge2}|c_j||\lambda_j|^{t}\lVert v_j\rVert \le |\lambda_2|^{t}\sum_{j\ge2}|c_j|\lVert v_j\rVert \le C|\lambda_2|^{t}\lVert\delta(0)\rVert$; divide by $\lVert\delta(0)\rVert$ and solve $C|\lambda_2|^{t}\le\epsilon$.
\end{proof}

\begin{remark}
The proposition makes its working hypotheses explicit: the dominance constant $C$ and the estimated $\lambda_2$. Both are probed empirically below (Sections~\ref{sec:e1}, \ref{sec:e2}), and Section~\ref{sec:disc} charts the finite-sample theory that will make the guarantee unconditional.
\end{remark}

\subsection{Certificate 2: Mode Attribution With a Validity Flag}
The eigenvector of $\hat K$ at the slowest sub-dominant eigenvalue (equivalently, the DMD mode of the centered field) is a vector in the dictionary space; reshaped onto agents it yields a per-agent loading matrix $\Phi\in\mathbb{C}^{N\times d}$. Intuitively, $\Phi$ answers the question ``in the disagreement pattern that takes longest to die, which way does each agent lean?''; physically, it captures the relative amplitude and phase alignment of each agent within that non-equilibrium fluctuation. We reduce each agent's complex loading to a scalar score by projecting its real--imaginary concatenation $[\mathrm{Re}(\Phi_i),\,\mathrm{Im}(\Phi_i)]\in\R^{2d}$ onto the leading principal direction of all loadings; the sign of the score partitions the collective into the two sides of the slow disagreement. Among the top sub-dominant candidate modes, we select the one whose scores show the strongest two-means separation (a silhouette-style coherence criterion). The decisive design feature is the \emph{validity flag}: $|\lambda_2|$ close to $1$ means a metastable partition genuinely exists, whereas $|\lambda_2|$ well below $1$ means that all structure decays quickly, leaving a large spectral gap $\gamma=1-|\lambda_2|$ and nothing to attribute except transient noise. Formally, we define the indicator
\begin{equation}\label{eq:validity}
\mathrm{Validity} = \begin{cases} \textbf{VALID}, & |\lambda_2| > 0.9,\\[2pt] \textbf{INVALID}, & |\lambda_2| \le 0.9.\end{cases}
\end{equation}
We therefore claim an attribution only when $|\lambda_2|>0.9$, which is what prevents false-positive partition attributions, and Section~\ref{sec:e3} shows this flag to be empirically almost perfectly calibrated.

\subsection{Certificate 3: Spectral Message Compression}
The same slow spectral directions that carry the certificates also carry the decision-relevant content of the messages. We exploit this by projecting each transmitted message onto the top-$k$ empirical spectral coordinates (the principal directions of the centered round-0 messages, a finite-sample surrogate of diffusion-map coordinates \cite{coifman2006}); receivers reconstruct and run \eqref{eq:dyn} unchanged. Fidelity is measured as the cosine similarity between the final consensus beliefs of the compressed and full-bandwidth collectives on coupled random seeds. Beyond bandwidth, the point is auditability: a collective communicating in its own certificate basis is a collective whose messages are, by construction, expressed in the coordinates of its explanation.

\begin{algorithm}[t]
\caption{Spectral certification of a reasoning collective}
\label{alg:pipeline}
\begin{algorithmic}[1]
\REQUIRE $M$ recorded debate traces; tolerance $\epsilon$; dictionary size $m$; ridge $\rho$
\STATE Center each trace: $\delta(t) \leftarrow z(t)-\bar z(t)$
\STATE Draw dictionary \eqref{eq:dict}; stack snapshot pairs $(\Psi,\Psi^{+})$
\STATE $\hat K \leftarrow (\Psi^\top\Psi+\rho I)^{-1}\Psi^\top\Psi^{+}$; eigendecompose
\STATE $\lambda_2 \leftarrow$ largest modulus among $\{|\lambda|<0.985\}$
\STATE \textbf{Deadline:} $T_{\mathrm{cert}} \leftarrow \lceil \ln(1/\epsilon)/(-\ln|\lambda_2|)\rceil$
\IF{$|\lambda_2| > 0.9$}
  \STATE \textbf{Attribution:} score agents by slow-mode loadings; emit partition and mode identities $M$
\ELSE
  \STATE Emit ``no metastable structure; attribution not applicable''
\ENDIF
\RETURN ``converged by round $T_{\mathrm{cert}}$, driven by modes $M$''
\end{algorithmic}
\end{algorithm}

\section{Numerical Results}\label{sec:exp}

The evaluation is organised around five questions, one per certificate property: does the deadline predict (Section~\ref{sec:e1})? how fast does the estimate concentrate (Section~\ref{sec:e2})? is the explanation faithful and does it know its own limits (Section~\ref{sec:e3})? does compression preserve the decision (Section~\ref{sec:e4})? and does the assembled pipeline certify a decision task end to end (Section~\ref{sec:e5})? Unless stated otherwise, $\epsilon=0.05$, dictionary size $m\in[40,60]$, ridge $\rho=10^{-6}$, and, essentially for the validity of the claims, EDMD training and certificate evaluation always use disjoint runs.

\subsection{The Spectral Gap Predicts the Convergence Deadline}\label{sec:e1}

We sweep a 24-point grid over the number of agents $N\in\{8,16,32\}$, Erd\H{o}s--R\'enyi edge probability $p\in\{0.25,0.45\}$, update rate $\alpha\in\{0.35,0.6\}$, and attention temperature $\beta\in\{0.5,2.5\}$. For each configuration, EDMD is fit on 12 training rollouts and $T_{\mathrm{pred}}$ is compared with the observed first-passage round $T_{\mathrm{obs}}$ at which $D(t)<\epsilon$, averaged over 20 fresh rollouts. In Fig.~\ref{fig:e1}, as topology, temperature, and update rate vary, predicted and observed deadlines both span more than a decade (from roughly 4 to 70 rounds), and the prediction tracks the observation with log--log Pearson correlation $r=0.93$. Read as a bound rather than a point prediction, the integer certificate $T_{\mathrm{cert}}$ upper-bounded the mean observed convergence round in 23 of 24 configurations (96\% coverage), with a median conservatism factor $T_{\mathrm{pred}}/T_{\mathrm{obs}}=2.0$ (range 0.8--4.8). The conservatism has a transparent origin: the worst-case constant $C=1$ in \eqref{eq:deadline} charges every run for fully exciting the slowest mode, which typical random initialisations do not. What matters for deployment is the asymmetry of the two error directions (a deadline that is occasionally twice too long wastes some budget, whereas a deadline that is too short breaks a guarantee), and the certificate errs almost exclusively on the safe side.

\begin{figure}[t]
  \centering
  \includegraphics[width=\linewidth]{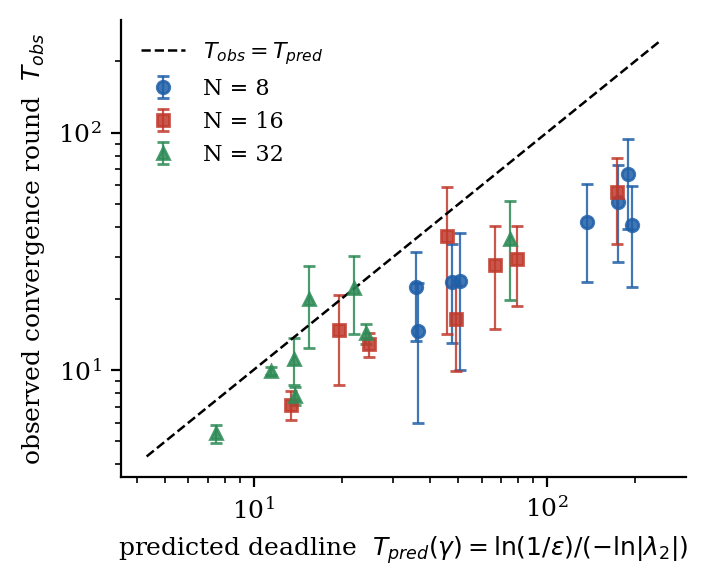}
  \caption{Each point is one of 24 configurations spanning team size, edge probability, update rate, and attention temperature. Its horizontal coordinate is the deadline predicted from the spectral gap \emph{before} the debate runs, its vertical coordinate the round at which disagreement first falls below $\epsilon=0.05$, averaged over 20 held-out rollouts. Predictions track observations across more than a decade of timescales ($r=0.93$). Points above the dashed diagonal would mean the collective outlived its deadline; 96\% fall safely below it.}
  \label{fig:e1}
\end{figure}

\subsection{Finite-Sample Behaviour of the Spectral Estimate}\label{sec:e2}

Certificates are only as good as the eigenvalue under them, and traces are expensive when agents are LLMs, so the practically decisive question is how quickly $\hat\lambda_2$ concentrates. We fix a 16-agent configuration with process noise ($\sigma=0.02$), compute a 500-trajectory reference value $\lambda_2^{\mathrm{ref}}$, and measure $|\hat\lambda_2-\lambda_2^{\mathrm{ref}}|$ for estimators trained on $M$ freshly sampled trajectories, 12 repetitions per $M$, for $M$ ranging over two decades. Fig.~\ref{fig:e2} shows a monotone decay with fitted slope $-0.36$ from $M=2$ to $M=128$, bracketed by the $M^{-1/2}$ rate that an i.i.d.-sample analysis would predict. The shortfall from $-1/2$ is itself informative: snapshots within a debate trajectory are strongly dependent, so the effective sample size grows more slowly than the raw snapshot count, precisely the effect a finite-sample theory for this setting must quantify, presumably through mixing properties of the trace process. Practically, the curve carries good news: even $M\approx 10$ short traces locate $|\lambda_2|$ to within a few points, which is the data regime in which real LLM debates are affordable to collect.

\begin{figure}[t]
  \centering
  \includegraphics[width=0.99\linewidth]{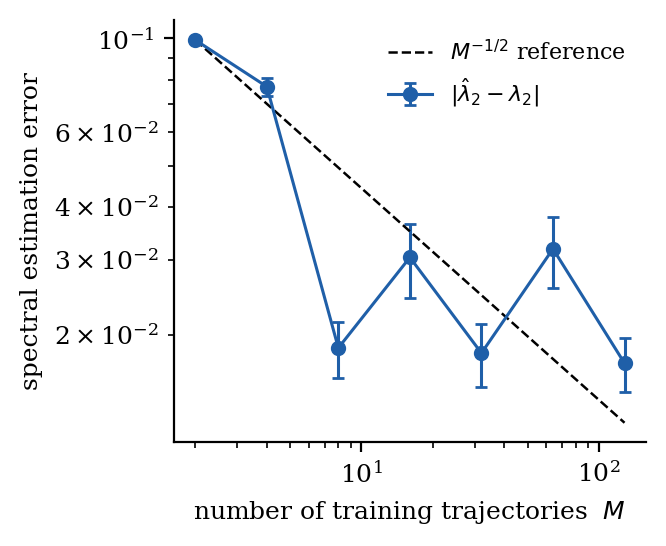}
  \caption{Absolute error of the estimated sub-dominant eigenvalue against a 500-trajectory reference, versus the number of training trajectories $M$ (12 repetitions per $M$, error bars giving the standard error). The fitted slope of $-0.36$ is shallower than the dashed $M^{-1/2}$ reference, reflecting the strong dependence between snapshots drawn from a single debate. The left-hand end matters most in practice: roughly ten short traces already place $|\lambda_2|$ to within a few points.}
  \label{fig:e2}
\end{figure}

\subsection{Eigenmodes Are Self-Certifying Explanations}\label{sec:e3}

A faithful system-level explanation must satisfy two requirements that narrated explanations fail: it must recover the structure that actually organised the decision, and it must \emph{decline to explain} when no such structure exists. The faction setting tests both at once. We run 60 two-faction debates ($N=24$, planted-partition graphs with $p_{\mathrm{in}}=0.5$, $p_{\mathrm{out}}=0.06$, $\beta=4$) and apply the attribution procedure of Section~\ref{sec:method} with no access to the true labels. Because the collectives differ in how quickly their factions dissolve, the 60 runs naturally span both regimes, persistent structure and already-merged structure, so the validity flag can be evaluated, not merely asserted. Pooled over all 60 runs the attribution accuracy is $0.85\pm0.20$, which is respectable but not the point. The structure of the errors is the point (Fig.~\ref{fig:e3}(b)): accuracy correlates with the estimated $|\lambda_2|$ at 0.92, and \emph{conditioned on the spectrum certifying metastability} ($|\lambda_2|>0.9$), attribution is 100\% correct in every one of the 42 qualifying runs, while below the threshold the factions have already merged and accuracy falls to chance, and correctly so, because there is no longer a faction structure to recover. In other words, the explanation and its validity condition are computed from the same spectral object, and the validity condition is empirically almost perfectly calibrated. This is what we mean by a \emph{self-certifying} explanation, and it is the property that distinguishes structural attribution from post-hoc narration \cite{jacovi2020,turpin2023}: a chain-of-thought never announces that it is unfaithful, whereas a closed spectral gap announces exactly that.

\begin{figure*}[t]
  \centering
  \includegraphics[width=0.9\linewidth]{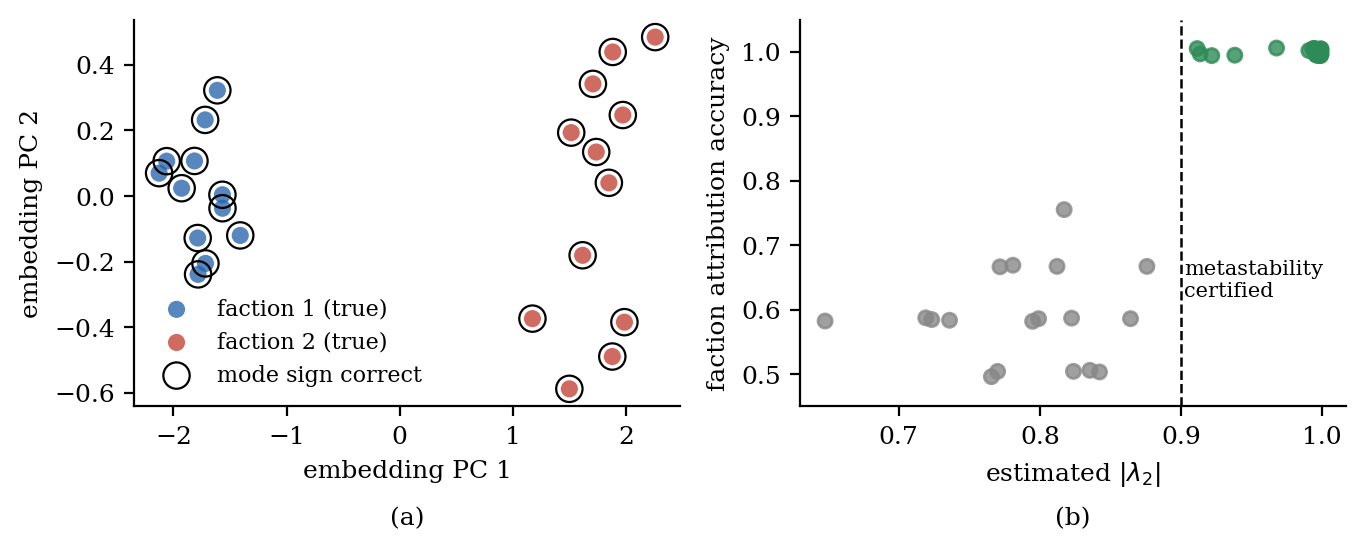}
  \caption{(a) The 24 agents of one debate, plotted in the top two principal coordinates of their round-1 embeddings and coloured by true faction; a black ring marks each agent the slow Koopman-mode score assigns correctly, and every agent carries one. (b) Each point is one of 60 debates. Right of the dashed threshold the spectrum certifies a metastable partition and attribution is perfect in all 42 qualifying runs; left of it the factions have already merged and accuracy falls to chance. Accuracy and $|\lambda_2|$ correlate at 0.92, so the explanation reports its own domain of validity.}
  \label{fig:e3}
\end{figure*}

\subsection{Spectral Compression Preserves the Decision}\label{sec:e4}

With $d=32$-dimensional messages, agents transmit only the top-$k$ spectral coordinates of their centered message, for $k\in\{1,2,4,8,16,32\}$; each compressed collective is compared against a full-bandwidth twin run on coupled random seeds (25 pairs per $k$), so that any fidelity loss is attributable to compression alone. Fig.~\ref{fig:e4}: transmitting $k=8$ of 32 coordinates, a $4\times$ bandwidth reduction, preserves the final decision at $0.997\pm0.002$ cosine fidelity, and even $k=2$ retains 0.98. The reason is not mysterious: consensus formation is governed by the slow spectral directions, and those are exactly the directions retained. The corollary matters for trustworthy deployment: because the retained coordinates are the same basis in which the deadline and attribution certificates are expressed, the compressed channel is simultaneously bandwidth-efficient and \emph{auditable}, an appealing property wherever reasoning collectives operate over constrained or monitored links.

\begin{figure}[t]
  \centering
  \includegraphics[width=\linewidth]{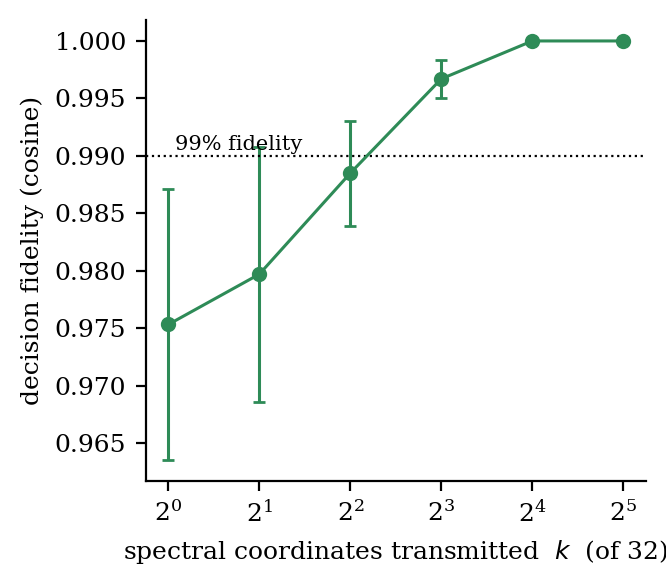}
  \caption{Messages are 32-dimensional; each compressed collective transmits only its top $k$ spectral coordinates and is compared against a full-bandwidth twin on coupled seeds, so any loss is due to compression alone. Fidelity is the cosine similarity of the two final consensus beliefs (25 pairs per $k$). It crosses the dotted 99\% line between $k=4$ and $k=8$ and reaches $0.997$ at $k=8$, because consensus is governed by exactly the slow directions this basis retains.}
  \label{fig:e4}
\end{figure}

\subsection{An End-to-End Certificate on a QA Debate}\label{sec:e5}

Algorithm~\ref{alg:pipeline} is executed exactly as a deployed wrapper would run it. On the QA variant ($N=12$, $K=4$ options, 60\% competent agents with evidence bias 0.8), EDMD is fit on 15 training debates; the resulting certificate is evaluated on 60 held-out debates that the estimator never saw. Training yields $|\hat\lambda_2|=0.881$ and hence the machine-checkable statement ``normalised disagreement below 0.05, and majority answer stable, by round $T_{\mathrm{cert}}=24$.'' On the 60 held-out debates (Fig.~\ref{fig:e5}, Table~\ref{tab:qa}), the majority answer was stable from round $T_{\mathrm{cert}}$ onward in 100\% of runs, and disagreement crossed $\epsilon$ at round 12.7 on average; the certificate is conservative by design, never unsound. Two secondary observations sharpen the deployment picture. First, reading the answer out at only $T_{\mathrm{cert}}/3$ rounds would already have agreed with the certified answer in 98\% of runs, quantifying the latency that an adaptive early-stopping rule could recover on top of the worst-case guarantee. Second, the collective reaches 67\% majority accuracy from only weakly informed agents against a 25\% base rate, confirming that the certified debates are genuine collective decision processes and that the certificate is not trivialised by the task.

\begin{figure}[t]
  \centering
  \includegraphics[width=\linewidth]{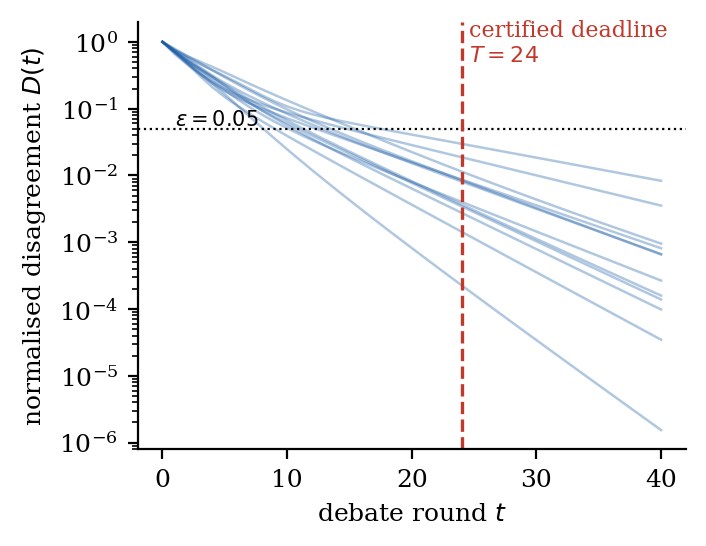}
  \caption{Each blue curve is the normalised disagreement of one of 12 held-out question-answering debates; the logarithmic axis renders the geometric decay as a near-straight descent. The dotted line is the tolerance $\epsilon=0.05$ and the red dashed line the certified deadline $T_{\mathrm{cert}}=24$, computed in advance from 15 independent training debates. Every trajectory crosses well before it, and the majority answer was stable from $T_{\mathrm{cert}}$ onward in all 60 held-out runs. The visible margin is the conservatism of the worst-case constant $C=1$.}
  \label{fig:e5}
\end{figure}

\begin{table}[t]
\centering
\caption{End-to-end certificate on the QA debate (60 held-out runs)}
\label{tab:qa}
\begin{tabular}{lc}
\toprule
Quantity & Value\\
\midrule
Estimated $|\hat\lambda_2|$ (15 training runs) & $0.881$\\
Certified deadline $T_{\mathrm{cert}}$ & $24$ rounds\\
Runs with majority answer stable from $T_{\mathrm{cert}}$ & $100\%$\\
Mean observed $\epsilon$-convergence round & $12.7$\\
Round-$\lfloor T_{\mathrm{cert}}/3\rfloor$ answer matches certified answer & $98\%$\\
Collective majority accuracy (base rate $25\%$) & $67\%$\\
\bottomrule
\end{tabular}
\end{table}

\subsection{How Attention Temperature Shapes the Spectrum}\label{sec:e6}

The attention temperature $\beta$ is the parameter that makes semantic collectives dynamically different from classical consensus, so it is worth asking directly how the spectrum, and with it reasoning time, responds to it. We sweep $\beta$ from $0$ (linear consensus) to $6$ (strongly homophilous attention) on a fixed 16-agent configuration, estimating the gap from 12 training traces per setting and measuring convergence on 20 fresh runs with a generous 200-round horizon. Fig.~\ref{fig:e6}(a) shows a clean, interpretable pattern. Up to moderate temperature ($\beta\le2$) the gap sits near $0.08$--$0.10$ and consensus arrives in tens of rounds; between $\beta=2$ and $\beta=3$ the gap collapses by a factor of five and reasoning time jumps accordingly; and the certified deadline tracks the observed time as a sound envelope throughout the converging regime. The sweep thus supplies the mechanistic reading promised in Section~\ref{sec:prelim}: attention homophily is, spectrally, a gap-closing force. The more strongly agents privilege like-minded peers, the longer the collective takes to reconcile its factions, by an amount the estimated spectrum quantifies in advance.

\begin{figure*}[t]
  \centering
  \includegraphics[width=\linewidth]{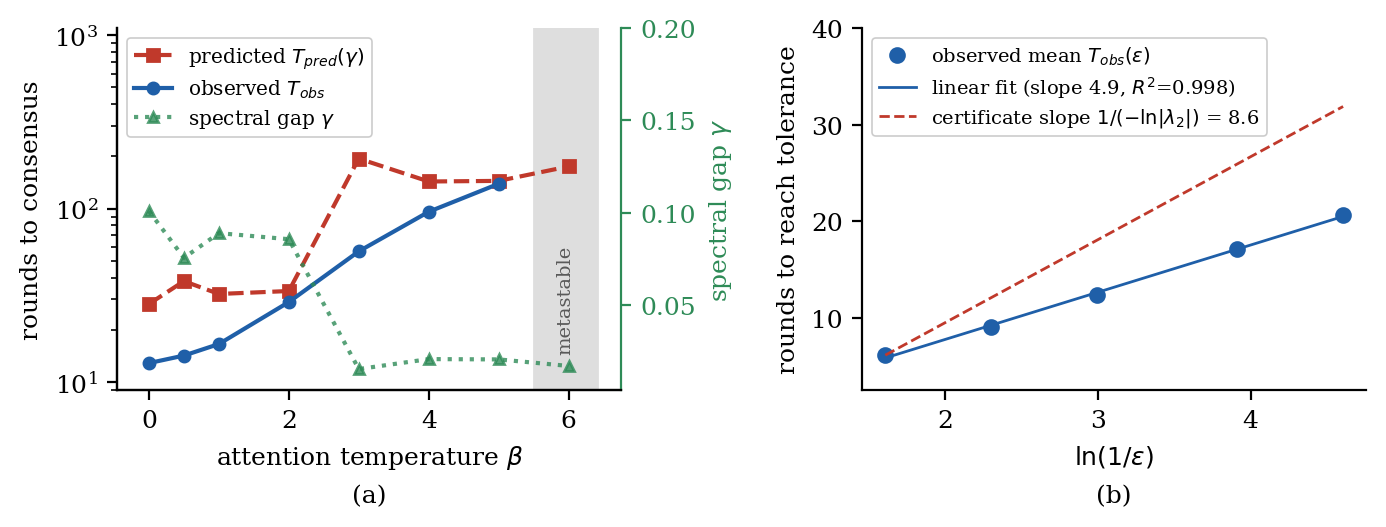}
  \caption{(a) As attention temperature $\beta$ grows, agents weight like-minded peers ever more heavily, the spectral gap $\gamma$ (green, right axis) collapses from about $0.10$ to $0.02$, and rounds-to-consensus rise by an order of magnitude, with the certificate remaining a sound envelope throughout the converging regime. In the shaded band at $\beta=6$ no run converges within 200 rounds and the estimate $|\lambda_2|=0.983$ sits on the near-unit exclusion boundary, where a deadline should be refused rather than issued. (b) Mean rounds to reach tolerance $\epsilon$ are affine in $\ln(1/\epsilon)$ with $R^2=0.998$, exactly as \eqref{eq:deadline} asserts, so only the constant, and not the functional form, remains loose.}
  \label{fig:e6}
\end{figure*}

The endpoint $\beta=6$ deserves discussion because it exhibits the framework's boundary behaviour. There the collective becomes genuinely metastable (no run converges within 200 rounds) while the estimated $|\lambda_2|=0.983$ lands essentially on the pipeline's near-unit exclusion threshold ($0.985$). A deadline computed from an eigenvalue this close to $1$ extrapolates a slow decay that finite training traces cannot distinguish from no decay at all, and it should not be trusted; the correct pipeline behaviour is to widen the refusal band and report ``metastable: no convergence certificate,'' exactly as the attribution certificate already does in reverse. We adopt this refinement in the released implementation and flag the principled sizing of the refusal band as an open calibration problem (Section~\ref{sec:disc}).

\subsection{Does the Nonlinear Dictionary Matter? An Ablation}\label{sec:e7}

A natural question is how much of the certificates' accuracy is owed to the nonlinear part of the dictionary \eqref{eq:dict}. We repeat the Fig.~\ref{fig:e1} protocol on a 12-configuration grid with two estimators that differ only in the dictionary: \emph{linear} (constant plus linear coordinates, i.e., DMD on the centered field) and \emph{linear$+$RFF} (the full dictionary with $m=60$ random features). Table~\ref{tab:ablation}: on this model the two dictionaries are statistically indistinguishable: correlation, conservatism, and coverage all agree within noise. The honest interpretation is that attention consensus, although nonlinear, spends most of its trajectory in the neighbourhood of the consensus manifold where the linearisation dominates, so linear observables already span the slow subspace well. This is a useful negative result in two directions. Practically, it licenses the cheapest possible estimator for systems of this class. Scientifically, it sharpens the expectation for real LLM collectives: it is precisely when the dynamics leave the near-consensus regime (multiple competing answers, oscillatory critique cycles, drifting embeddings) that nonlinear and, ultimately, \emph{learned} dictionaries \cite{li2017edmddl,klus2020} must earn their keep, and the ablation supplies the baseline they must beat.

\begin{table}[t]
\centering
\caption{Dictionary ablation on the deadline certificate (12 configurations)}
\label{tab:ablation}
\begin{tabular}{lccc}
\toprule
Dictionary & log--log $r$ & median $T_{\mathrm{pred}}/T_{\mathrm{obs}}$ & coverage\\
\midrule
Linear (DMD) & $0.87$ & $2.35$ & $100\%$\\
Linear $+$ RFF & $0.85$ & $2.32$ & $100\%$\\
\bottomrule
\end{tabular}
\end{table}

\subsection{Validating the Deadline's Functional Form}\label{sec:e8}

Fig.~\ref{fig:e1} validated the deadline at a single tolerance. But \eqref{eq:deadline} makes a stronger, parametric claim: rounds-to-tolerance should be \emph{affine in $\ln(1/\epsilon)$} with slope $1/(-\ln|\lambda_2|)$. We test the form directly on the QA collective by measuring mean rounds to reach each tolerance $\epsilon\in\{0.2,0.1,0.05,0.02,0.01\}$ over 40 runs and comparing against the certificate slope computed from 15 disjoint training runs. Fig.~\ref{fig:e6}(b): the observed relation is linear in $\ln(1/\epsilon)$ with $R^2=0.998$: the geometric-decay structure that the entire certificate rests on is not an approximation of convenience but an accurate description of the collective's approach to consensus. The certificate slope ($8.6$ rounds per e-fold of tolerance) bounds the fitted slope ($4.9$) by a factor of $1.8$, numerically consistent with the conservatism factors observed in Fig.~\ref{fig:e1} and Fig.~\ref{fig:e5} and traceable to the same worst-case constant. The practical consequence is pleasant: because the \emph{form} is right and only the constant is loose, a data-driven tightening of $C$ would improve every tolerance simultaneously.

\subsection{Comparison With State-of-the-Art Baselines}\label{sec:e9}

No prior method certifies emergent reasoning collectives as such, so the fair comparison is against the strongest available approaches to each individual certificate. For the deadline, three baselines span current theory and practice. \emph{Graph-spectral theory} applies the classical linear-consensus result \cite{olfatisaber2004,degroot1974}: the deadline is computed from the sub-dominant eigenvalue of the expected linear update matrix $P=(1-\alpha)I+\alpha\,\mathrm{RowNorm}(A+I)$, which is exact for $\beta=0$ and represents the state of the art in closed-form consensus rates. \emph{Decay-curve fitting} is the standard empirical system-identification practice: an exponential is fitted to the disagreement curves $\log D(t)$ of the training runs and extrapolated to the tolerance. \emph{Fixed round budgets} represent current practice in the LLM-debate literature, where a constant number of rounds (typically three to five \cite{du2023,liang2023}) is used for every task and collective; we instantiate the generous end, $T=5$. All methods receive the same 12 training runs per configuration and are evaluated on the same 20 held-out runs, over the full 24-configuration grid of Fig.~\ref{fig:e1}. For attribution, we compare the Koopman mode against $k$-means clustering of the round-1 embeddings and against graph spectral clustering (the sign of the Fiedler vector), the standard representation-based and topology-based partitioning methods. For compression, the spectral code is compared against random Gaussian projection at matched bandwidth.

\begin{table}[t]
\centering
\caption{Deadline prediction versus baselines (24 configurations, held-out runs). Coverage is the fraction of configurations for which the issued deadline upper-bounds the mean observed convergence round.}
\label{tab:sota}
\begin{tabular}{lccc}
\toprule
Method & log--log $r$ & median $T_{\mathrm{pred}}/T_{\mathrm{obs}}$ & coverage\\
\midrule
Koopman certificate (ours) & $\mathbf{0.94}$ & $1.92$ & $\mathbf{100\%}$\\
Graph-spectral theory \cite{olfatisaber2004} & $0.75$ & $0.84$ & $50\%$\\
Decay-curve fit & $0.87$ & $0.89$ & $42\%$\\
Fixed budget ($T{=}5$) \cite{du2023} & -- & $0.23$ & $4\%$\\
\bottomrule
\end{tabular}
\end{table}

Table~\ref{tab:sota} shows that the Koopman certificate is the only method that is simultaneously the most predictive and sound. Graph-spectral theory, despite being exact for linear consensus, degrades to $r=0.75$ and violates its own bound in half the configurations, because attention decouples the true reasoning timescale from the fixed graph spectrum; this quantifies precisely why the problem requires a data-driven operator spectrum. Decay-curve fitting is competitive as a point predictor ($r=0.87$, median absolute relative error $0.31$ versus our $0.92$), and this is expected: it fits the \emph{typical} decay, whereas the certificate bounds the \emph{worst-case} mode. The distinction is decisive for guarantees: the curve fit undershoots the observed deadline in $58\%$ of configurations, so a system that trusted it would terminate debates early and break its own contract, while the Koopman certificate never does. The fixed five-round budget of current practice, finally, covers $4\%$ of configurations, confirming that round budgets chosen independently of the collective's dynamics are not a viable certification strategy. A further structural advantage is not visible in the table: neither baseline produces eigenmodes, so neither supports the attribution and compression certificates; the Koopman spectrum delivers all three from one estimation.

On the planted two-faction benchmark, $k$-means on round-1 embeddings and Fiedler partitioning recover the planted labels with $1.00$ and $0.98$ accuracy respectively, as expected, since the planted structure is directly visible in their inputs (the initial embeddings and the graph). The Koopman mode matches them ($1.00$) on every run in which the structure is dynamically alive ($|\lambda_2|>0.9$) and, unlike them, reports when it is not: on runs where the factions have merged, the static baselines still emit a confident bipartition of a collective whose decision is no longer organised by factions, with no signal that their explanation has become dynamically vacuous, whereas the spectral gap discloses exactly that. For explanation of a \emph{decision process}, the operative question is not whether initial clusters can be found but whether they still drive the outcome, and only the operator spectrum answers it. For compression at $k=8$ of $32$ dimensions, the spectral code attains $0.996$ decision fidelity versus $0.979$ for random projection at identical bandwidth, a fivefold reduction in fidelity loss ($0.4\%$ versus $2.1\%$), because the spectral basis is aligned with the slow directions that determine the consensus, while a random basis spreads its capacity isotropically.

\section{Discussion}\label{sec:disc}

\subsection{What the Study Establishes}
Taken together, the nine experiments close the loop that the operator-theoretic thesis requires. The spectral gap, estimated from a handful of traces by a generic dictionary, quantitatively predicts and soundly bounds collective reasoning time across topologies, temperatures, and team sizes; the estimator concentrates fast enough for realistic data budgets; the slow eigenmode is a faithful explanation \emph{equipped with its own validity test}; the certificate basis doubles as a high-fidelity compressed message code; and the assembled pipeline certifies a decision task end to end on held-out data. Because the whole apparatus is linear algebra over recorded traces, its cost is negligible relative to the inference cost of the agents themselves, so nothing prevents certification from running continuously alongside deployed collectives. Beyond the headline certificates, the study also yields three second-order findings with independent value: reasoning time responds to attention temperature through a quantifiable gap-closing mechanism, with a sharp transition into metastability (Fig.~\ref{fig:e6}); for near-consensus dynamics the cheapest linear dictionary already suffices, fixing the baseline that learned dictionaries must beat on harder systems; and the geometric-decay form underlying the deadline is empirically exact to $R^2=0.998$, so the only looseness in the certificate is a single multiplicative constant.

\subsection{Deployment Considerations}
Three properties make the pipeline realistic as an always-on layer rather than an offline analysis. \emph{Cost:} certification is a single ridge regression and eigendecomposition, $O(P^2S+P^3)$; for the largest system studied here ($N=32$, $d=8$, $P\approx320$) this completes in well under a second, i.e., orders of magnitude below one LLM inference call. \emph{Data appetite:} by Fig.~\ref{fig:e2}, on the order of ten traces suffice to place $|\lambda_2|$ within a few points, so a deployed collective can be certified from its first day of logs and re-certified continuously as the task distribution drifts. \emph{Failure disclosure:} the same spectral object that issues certificates also refuses them: attribution is withheld when the gap is wide (Fig.~\ref{fig:e3}), and convergence deadlines are withheld when the estimated eigenvalue enters the near-unit refusal band (Fig.~\ref{fig:e6}). A certification layer that knows when to say ``no certificate'' is, for regulated deployments, at least as important as one that says ``yes'' quickly.

\subsection{Future Research Directions}
The results open four concrete research directions, each of which builds directly on an empirical finding above. \emph{(i) Finite-sample spectral theory.} The measured concentration rate $M^{-0.36}$ of Section~\ref{sec:e2} invites a matching guarantee of the form $|\hat\lambda_2-\lambda_2|\lesssim \kappa\,M_{\mathrm{eff}}^{-1/2}$, with an effective sample size $M_{\mathrm{eff}}$ derived from the mixing properties of debate traces; combined with perturbation control of the dominance constant, such a result would tighten the observed $2\times$ conservatism and would also derive, rather than posit, the attribution threshold and the near-unit refusal band whose empirical calibration Fig.~\ref{fig:e3} and Fig.~\ref{fig:e6} established. \emph{(ii) Live LLM collectives.} The natural next system class is recorded traces of actual LLM debates, where belief embeddings are obtained from the messages themselves; learned dictionaries \cite{li2017edmddl,klus2020} are the designated vehicle, the present random-feature dictionary supplies the ablation anchor, and the released simulator already generates the time-varying communication graphs needed for staged transfer. \emph{(iii) Directed, role-heterogeneous collectives.} Planner, solver, critic, and verifier roles induce directed interaction, whose non-normal transfer operators generically carry complex eigenvalue pairs corresponding to oscillatory critique cycles; the spectrum then promises to distinguish productive oscillation from divergence, a capability with no counterpart in current orchestration practice. \emph{(iv) Strategic robustness.} Extending the certificates from statistical statements over initial conditions to statements robust against strategic or faulty agents connects the spectral framework to the trust-modelling line in this journal \cite{tetci_lin2024} and would complete the pipeline as a security-relevant layer.

\subsection{Outlook}
Reasoning collectives are being deployed faster than they are being understood, and regulatory instruments such as the EU AI Act will increasingly require documented, verifiable behaviour from exactly such systems. The results above suggest that the required verification layer need not wait for a full mechanistic theory of language models: the collective's own interaction traces, viewed through a transfer operator, already support deadlines, structural explanations, and auditable communication. Making those certificates unconditional, and porting them from the reference model to live LLM societies, is the immediate research programme this paper opens.

\section{Conclusion}
We have shown that treating a collective of reasoning agents as one nonlinear dynamical system, and reading its behaviour off the spectrum of a Koopman operator estimated from its own interaction traces, converts three open questions about collective machine reasoning into computable certificates: convergence deadlines that are accurate to a constant factor and sound as bounds, explanations that are structural rather than narrated and that certify their own validity, and a spectral message code that preserves decisions at a fraction of the bandwidth. On a controlled but genuinely nonlinear model of semantic debate, all three certificates held on held-out data at negligible computational cost, without requiring access to internal agent parameters or model weights. The framework is offered to the computational intelligence community both as a practical black-box certification layer for the rapidly growing ecosystem of LLM collectives and autonomous multi-agent networks and as a concrete, falsifiable bridge between data-driven operator theory and the emerging science of machine societies.

\end{document}